\documentclass[12pt]{article}

\usepackage{amsfonts}
\usepackage{amssymb}
\usepackage{amsmath}

\usepackage{graphicx}
\usepackage{latexsym}
\usepackage{color}
\usepackage{authblk}

\usepackage{graphicx,tikz}
\usepackage{amsmath,amssymb,mathrsfs}

\usepackage{amsthm}

\newcommand{\bm}[1]{\mbox{\boldmath $#1$}}

\def\be{\begin{equation}}
\def\ee{\end{equation}}
\def\bea{\begin{eqnarray}}
\def\eea{\end{eqnarray}}
\def\bean{\begin{eqnarray*}}
\def\eean{\end{eqnarray*}}

\def\square{\hfill\hbox{\vrule\vbox{\hrule\phantom{N}\hrule}\vrule}\,}

\newtheorem{theorem}{Theorem}[section]
\newtheorem{result}{Result}[section]

\newtheorem{defi}{Definition}
{\theoremstyle{definition}
\newtheorem{remark}{Remark}[section]}

\def\proof{\noindent{\em Proof.\/}\hspace{3mm}}

\newlength{\cellwidth}
\begin{document}
\title{New bounds for the area of MOTS and generalized ultra-massive spacetimes}

\author{Jos\'e M. M. Senovilla$^{1,2,3}$}
\affil{$^1$Departamento de F\'isica, Universidad del Pa\'is Vasco UPV/EHU, Apartado 644, 48080 Bilbao, Spain\\
$^2$ EHU Quantum Center, Universidad del Pa\'{\i}s Vasco UPV/EHU.\\
$^3$ Corresponding author: josemm.senovilla@ehu.eus}

{\let\newpage\relax\maketitle}

\vspace{-0.2em}

\begin{abstract}
Bounds for the area of general closed marginally trapped surfaces (MTSs) are presented. They do not require any stability condition, and are determined by a constant that depends on a particular component of the Einstein tensor on the surface and another constant that governs the (in)stability of the MTS. When stability is imposed, the area bounds are refined. These bounds are realized in spacetimes exhibiting interesting generic properties: they possess marginally trapped tubes foliated by marginally trapped topological spheres containing a distinguished round sphere $\bar S$ with constant Gaussian curvature that saturates the area bound. This distinguished surface separates two distinct regions of the marginally trapped tubes: a dynamical horizon and a timelike membrane. The particular case where there is a positive cosmological constant leads to the well-known universal bound $4\pi/ \Lambda$ for spatially stable MTSs, and to the recently introduced `ultra-massive spacetimes'.  These spacetimes are more extreme than black holes, as there is no event horizon and the entire exterior region undergoes unavoidable collapse. In this paper similar behaviour is found for non-positive $\Lambda$ if the energy-momentum content is powerful enough. The results may have implications for binary mergers and on accreting very compact objects.
\end{abstract}

\section{Introduction}
Spatially stable closed marginally future-trapped\footnote{\label{foot1}All the results are applicable, {\em mutatis mutandis}, to marginally past-trapped surfaces, but for the sake of conciseness, only future-trapped surfaces will be considered  from the start.} surfaces (MTSs) are usually regarded as a hallmark of (dynamical) black holes (BHs). The qualifier ``spatially stable'' is crucial, as it guarantees the existence of spacelike perturbations towards the {\em exterior} of the MTS that lead to untrapped surfaces, while there are perturbations towards its interior resulting in (weakly) trapped surfaces \cite{AMS,AMS1,PBH}. The area of MTSs is commonly associated with the size and strength of the corresponding BH, and hence with its entropy, as it is the case in BHs in equilibrium \cite{Bek,Bek1,Haw,AK}.

In this work, a general bound for the area $A_S$ of any MTS $S$, whether stable or not, is presented, given explicitly by
\be\label{lim}
(\mu_S +\lambda_{-\ell}) A_S \leq 4\pi (1-g)
\ee
where $g$ is the genus of $S$, $\lambda_{-\ell}$ is a constant to be defined later associated with the stability of the MTS and the constant $\mu_S$ is defined as
\be\label{muS}
\mu_S = \min_S G_{\mu\nu} \ell^\mu k^\nu .
\ee
Here $G_{\mu\nu}$ is the Einstein tensor and $\ell^\mu, k^\mu$ are the two null normal vector fields orthogonal to $S$ normalized such that 
$$
\ell^\mu k_\mu =-1.
$$
(There remains a boost freedom $l^\mu \rightarrow A\l^\mu$ while $k^\mu \rightarrow k^\mu/A$ that obviously does not affect $\mu_S$). Previous bounds of this type were obtained in \cite{GM,G} for the case of MTSs {\em stable within a given initial hypersurface} that satisfies the dominant energy condition in General Relativity (GR). However, those bounds may fail to be optimal, as they only consider perturbations within the initial slice. As we show below, stronger bounds exist even when the surface is unstable within the initial data set, as explained in Remark \ref{slice}.

It is worth emphasizing that the area bound \eqref{lim} holds in any gravitational theory based on a Lorentzian manifold. In the important case of GR, using the Einstein field equations
\be\label{efe}
G_{\mu\nu}+\Lambda g_{\mu\nu} = \frac{8\pi G}{c^4} T_{\mu\nu}:={\cal T}_{\mu\nu} 
\ee
one has
\be\label{muinT}
\mu_S = \Lambda + \min_S {\cal T}_{\mu\nu} l^\mu k^\nu \hspace{1cm} \mbox{(only in GR)}
\ee
where $\Lambda$ is the cosmological constant.
Notice that the second summand on the righthand side is non-negative if, for instance, the dominant energy condition (DEC) holds.

At first sight, in the important case with $\lambda_{-\ell}=0$ the bound \eqref{lim} seems to lead to a maximum size (or entropy) of topologically spherical ($g=0$) BHs if $\mu_S >0$ ---and also to a minimum one for higher genus black holes if $\mu_S <0$, see for instance \cite{W}. However, as shown in \cite{Snew,Snew2} in the case where $\mu_S=\Lambda >0$
the corresponding area bound\footnote{This bound was found years ago in \cite{HSN} later discussed and improved in \cite{W,Simon}} is not a physical limitation for very large BHs of spherical topology, rather it leads to new types of collapsed objects, different and more powerful than BHs, termed {\em ultra-massive spacetimes}.

A similar situation arises for the general case where $\mu_S$ has any positive value: any dynamical horizon foliated by spherical  MTSs that approach indefinitely the area bound stops being spacelike to become a timelike membrane at the distinguished round sphere $\bar S$, which has constant Gaussian curvature ${\cal K} =\mu_{\bar S}$ and area $A_{\bar S}=4\pi/\mu_S$. In other words, the entire hypersurface foliated by MTSs constitutes a marginally trapped tube (MTT) \cite{AG,Booth,BBGV} that changes signature precisely at $\bar S$. These hypersurfaces are also called {\it generalized holographic screens} (GHS) \cite{BE,BE1,Snew,Snew2} because they satisfy an area law: the area of the foliating MTSs is always non-decreasing following the external directions within the MTT, where the concept {\em external} can be characterized in a precise and unique manner, see \cite{Snew2}. Any other possible MTT intersecting $\bar S$ changes signature somewhere on $\bar S$ too. Typically, a singularity arises, and a {\em generalized ultra-massive spacetime} is left over. They tend not to have any event horizon nor future null  infinity, i.e., they are {\em not black holes}. The key point is that any MTS that reaches the area bound \eqref{lim} with $\lambda_{-\ell}=0$ cannot be stable in any external spatial direction.

The terminology, notation and results of \cite{Snew2} will be repeatedly used, because the underlying ideas and proofs are basically the same. Reference to the theorems and formulas in \cite{Snew2} will also be made for the sake of conciseness and to avoid unnecessary repetitions. The main results are then easily derived. 

But before anything, let us fix the setup and the main concepts to be used.

\section{Definitions and the stability operator}
$(M,g)$ denotes a 4-dimensional oriented Lorentzian manifold with causally oriented metric $g_{\mu\nu}$ of signature $(-,+,+,+)$. 
No field equations are assumed in principle, so that the results can be applied to alternative theories based on Lorentzian geometry.

For the purposes of this work, a closed spacelike surface $S$ is any embedded and orientable 2-dimensional connected and compact manifold with positive-definite first fundamental form $h_{AB}$\footnote{Greek indices are spacetime indices and Latin upper-case indices belong to the 2-dimensional $S$}. $\mathfrak{X}(S)$ and $\mathfrak{X}^\perp(S)$ denote the set of vector fields tangent  and orthogonal to $S$, respectively. One can choose a couple of {\em null} vector fields in $\mathfrak{X}^\perp(S)$, denoted by $\vec \ell$ and $\vec k$, that are linearly independent all over $S$ \cite{Mars}. For convenience, they are chosen to be future pointing and normalized such that
$$
\ell_\mu k^\mu =-1.
$$
The two null second fundamental forms along $\vec \ell$ and $\vec k$ are denoted by $K^\ell_{AB}$ and $K^k_{AB}$ and the corresponding null shears and expansions by $\Sigma^\ell_{AB}$ and $\Sigma^k_{AB}$ and by $\theta^\ell$ and $\theta^k$ respectively. The mean curvature vector field of $S$ reads then

\be\label{H}
\vec H=-\theta^\ell \vec k -\theta^k \vec \ell
\ee

The causal character of $\vec H$ leads to a classification into different types of trapped surfaces, see e.g. \cite{AG,AMS,AMS1,MaSe,S,S0} and in particular \cite{Snew2} for the particular nomenclature used here. For our purposes, we just need to know (recall footnote \ref{foot1}) that $S$ is a marginally outer future trapped surface (MOTS) if $\vec H$ is aligned with one of the null normals all over $S$; and $S$ is said to be marginally future trapped (MTS) if in addition $\vec H$ is future pointing all over $S$.
Equivalently, MOTSs have one of the null expansions vanishing on $S$, and MTSs also have the other null expansion negative. Note that any MTS is also a MOTS, hence all results relative to MOTSs do apply to MTSs too.
To fix the notation, assume that $\vec k$ is the direction of the mean curvature vector, so that in this paper we are going to have $\theta^k =0$ with 
\be\label{H1}
\mbox{MOTS: } \vec H = -\theta^\ell \vec k, \hspace{1cm} \mbox{MTS}: \vec H = -\theta^\ell \vec k \, \, \, \& \, \, \, \theta^\ell <0
\ee
and the region in the normal plane at every $x\in S$ with positive component along $\vec k$ will be called the `external' region. This contains all the $\vec n\in \mathfrak{X}^\perp (S)$ such that $n^\mu k_\mu \geq 0$.

A marginally (outer) trapped tube M(O)TT is a hypersurface in $M$ foliated by M(O)TSs (recently they have also been termed as Quasi-local horizons \cite{AK}). A spacelike MTT is usually called a dynamical horizon \cite{AG,AK} or future outer trapping horizon \cite{Hay}, while if it is timelike is called a timelike  membrane \cite{AG,BBGV} or future inner trapping horizon \cite{Hay}.

\subsection{Stability of marginally outer trapped surfaces}\label{subsec:possibilities}
For the notion of stability of MOTSs one may consult \cite{AMS,AMS1,AK,AG,BeS,BBGV,BE,BE1,Hay,JRD,Mars,Ne,PBH,SPrague,SERE,SILS,Snew2,Simon} and references therein. The stability is defined with respect to a particular external normal direction $\vec n\in \mathfrak{X}^\perp(S)$ ---with $n_\mu k^\mu \geq 0$. One can visualize such directions by looking at  figure 1 in \cite{Snew2}. The precise definition is
\begin{defi}
A MOTS $S$ is stable along an external normal $\vec n\in \mathfrak{X}^\perp(S)$ if there is a non-negative function $f\not\equiv 0$ such that the variation $\delta_{f\vec n} \theta^k $ of $\theta^k$ along $f\vec n$ is non-negative. And $S$ is strictly stable if in addition $\delta_{f\vec n} \theta^k >0$ somewhere.
\end{defi}
The idea behind this definition is: a MOTS embedded in a hypersurface $\Sigma$ is strictly stable along the direction tangent to $\Sigma$ and orthogonal to $S$ if and only if no nearby (weakly) trapped surface reaches the `exterior' of $S$, and no nearby untrapped surface enters its interior, within $\Sigma$ \cite{AMS,AMS1}. This is a barrier property for MTSs that makes them the best candidates to signal local-in-time black-hole regions.

The basic tool needed in this work is the {\em stability operator} for MOTSs, as introduced in \cite{AMS,AMS1}.
The external pointing normal vector fields can be characterized by its norm 
\be
\vec n =-\vec\ell +\frac{n_{\mu}n^{\mu}}{2}\vec k  \label{n}
\ee
keeping its component along $\vec\ell$ fixed, or equivalently
\be
H_\mu n^\mu = -\theta^\ell \hspace{1cm} (k_\mu n^\mu=1) .\label{norm}
\ee
The causal character of $\vec n$ is unrestricted.

The variation of $\theta^k=0$ reads
\be
\delta_{f\vec n} \theta^k:=L_n f
\label{deltatheta}
\ee
where $L_n$ is the {\em stability operator} for $S$ in the direction $\vec n$ as computed in  \cite{AMS,AMS1}
\be
L_n f =-\Delta f+2s^B D_{B}f+f\left({\cal K} -s^B s_{B}+D_{B}s^B-G_{\mu\nu}k^\mu \ell^{\nu}-\frac{n^\rho n_{\rho}}{2}\,  W\right) \label{Ln} .
\ee
Here ${\cal K}$ is the Gaussian curvature on $S$, $\Delta$ its Laplacian, $D_B$ its covariant derivative, $s_{B}$ the one-form on $S$ defined by
$$
\bm{s} (\vec v) =k_{\mu}\nabla_{\vec v}\ell^\mu, \hspace{1cm} \forall \vec v\in \mathfrak{X}(S) 
$$
and finally
\be\label{W}
W:= G_{\mu\nu}k^\mu k^\nu + \Sigma^k_{AB} \Sigma^k{}^{AB} \geq G_{\mu\nu}k^\mu k^\nu .
\ee

For every $\vec n\in\mathfrak{X}^\perp (S)$, $L_n$ is an elliptic operator on $S$, not self-adjoint in general with respect to the $L^2$-product. The principal eigenvalue is {\em real}, denoted by $\lambda_n$, and the corresponding eigenfunction $\phi_n$
\be\label{phin}
L_n \phi_n =\lambda_n \phi_n, \hspace{3mm} \lambda_n \in \mathbb{R}, \hspace{4mm} \phi_n >0
\ee
can be chosen to be {\em real} and positive on all of $S$. 
The fundamental property proven in \cite{AMS,AMS1} is
\begin{result}
A MOTS is (strictly) stable along the external direction $\vec n$ if and only if $\lambda_n$ is non-negative (positive). 
\end{result}

\section{The area bound}\label{sec:bounds}
Substituting $f=\phi_n$ into \eqref{Ln}, dividing by $\phi_n$ and using \eqref{phin} one obtains
\be
\lambda_n = -\Delta \ln \phi_n +D_B s^B-\left(D_B \ln \phi_n -s_B\right)\left(D^B \ln \phi_n -s^B\right)+{\cal K} -\frac{n^\rho n_{\rho}}{2}W-G_{\mu\nu}k^\mu \ell^{\nu}\label{lambdan} .
\ee
Consider the last term on the righthand side and recall the constant $\mu_S$ defined in \eqref{muS}.
Then 
$$
\int_S G_{\mu\nu}k^\mu \ell^{\nu} =\mu_S A_S +E^2 \geq \mu_S A_S 
$$
for some non-negative constant $E^2$, where $A_S$ is the area of $S$. Integrating \eqref{lambdan} on the compact $S$ and using the Gauss-Bonnet theorem
\bea
\left(\mu_S +\lambda_n \right)A_S =
4\pi (1-g) -E^2 -\int_S\left[\left(D_B \ln \phi_n -s_B\right)\left(D^B \ln \phi_n -s^B\right)+ \frac{n^\rho n_{\rho}}{2}W \right]\label{genineq}
\eea
where $g$ is the genus of $S$. Therefore, if $(n^\rho n_\rho) W\geq 0$ on $S$, the integrand on the righthand side is non-negative and one obtains area bounds:
\begin{result}\label{areaGeneral}
Let $S$ be a MOTS and let $\vec n$ be a normal external direction such that $(n^\rho n_\rho)W\geq 0$ on the entire $S$. Then the following inequality holds
\be\label{GenIneq}
(\mu_S +\lambda_n) A_S \leq 4\pi (1-g) .
\ee
Equality holds if and only if in fact $(n^\rho n_\rho)W= 0$ on $S$ and 
\be\label{zeros0}
s_B =D_B \ln \phi_n, \hspace{1cm} G_{\mu\nu} k^\mu \ell^\nu \stackrel{S}{=} \mu_S =  \rm{constant}.
\ee
\end{result}
The proof follows immediately from \eqref{genineq} because $E^2=0$ requires the second equality in \eqref{zeros0} given that $\mu_S$ is the minimum of $G_{\mu\nu}k^\mu \ell^{\nu}$ on $S$.
\begin{remark}
Actually, only the non-negativity of the integral on $S$ of $(n^\rho n_\rho) W$ is required to get the result.
\end{remark}

Result \ref{areaGeneral}  applies to the null normal $\vec n =-\vec\ell$ and one derives the following important result.
\begin{result}\label{areaell}
Let $S$ be a MOTS. 
Then
\begin{enumerate}
\item\label{ineqnew} The following inequality holds
\be\label{ellineq}
\left(\mu_S +\lambda_{-\ell} \right)A_S\leq 4\pi (1-g) 
\ee
with equality if and only if
\be\label{zerosnew}
s_B =D_B \ln \phi_{-\ell}, \hspace{1cm} G_{\mu\nu} k^\mu \ell^\nu \stackrel{S}{=} \mu_S =\rm{constant}.
\ee
\item\label{ineq1new} In particular, if $S$ is stable ($\lambda_{-\ell}\geq 0$) along $-\vec\ell$ 
then we have
$$
\mu_S A_S\leq 4\pi (1-g) 
$$
equality requiring $\lambda_{-\ell}=0$ and \eqref{zerosnew}.
\item\label{topolnew} If $g>0$ and $\mu_S>0$ then necessarily $\lambda_{-\ell}<0$. In other words, if $\mu_S>0$, then stability along $-\vec \ell$ can only happen for topological spheres ($g=0$).
\item\label{Alargernew} If $\mu_S >0$ and the area $A_S > 4\pi /\mu_S$, then $\lambda_{-\ell} <0$ necessarily.
\item\label{Aequalnew} If $\mu_S >0$ and $A_S= 4\pi /\mu_S $, then $S$ can be stable but not strictly stable along $-\vec\ell$. In the stable case (so that $\lambda_{-\ell} = 0$), \eqref{zerosnew} hold and $S$ is a metric round sphere of constant Gaussian curvature ${\cal K}=\mu_S$.
\end{enumerate}
\end{result}

\proof The proof is analogous to that of Result 2.5 in \cite{Snew2}. Point \ref{ineqnew} is Result \ref{areaGeneral} applied to $-\vec\ell$ and point \ref{ineq1new} follows immediately from that. If the surface has genus $g=1,2,\dots$, then the righthand side of \eqref{ellineq} is non-positive giving point \ref{topolnew}. Point \ref{Alargernew} is included in point \ref{topolnew} unless for $g=0$, in which case this follows directly from \eqref{ellineq}. Point \ref{Aequalnew} is again non-trivial only for $g=0$, and then point \ref{ineqnew} implies that $\lambda_{-\ell}$ cannot be positive, the only stable possibility being given by $\lambda_{-\ell}=0$. Point \ref{ineq1new} gives then \eqref{zerosnew}. Finally, one goes back to \eqref{lambdan} applied to $\vec n =-\vec\ell$, which reads
\be
\lambda_{-\ell} = -\Delta \ln \phi_{-\ell} +D_B s^B-\left(D_B \ln \phi_{-\ell} -s_B\right)\left(D^B \ln \phi_{-\ell} -s^B\right)+{\cal K} -G_{\mu\nu}k^\mu \ell^{\nu}\label{lambdaell}
\ee
and inserting here \eqref{zerosnew} all terms except the last two cancel out leading to
$$
{\cal K} -\mu_S =0
$$
proving the result.
 \square
 
\begin{remark}\label{Sunstable}
It is important to remark that bounds for the area do exist even in cases where $S$ is {\em unstable} along $-\vec\ell$, as long as $\mu_S +\lambda_{-\ell} $ is positive. Of course, this requires that $\mu_S >0$.
\end{remark}

\subsection{Consequences under the assumption $W\geq 0$}\label{subsec:W>0}
Note that the function $W$ is non-negative if $G_{\mu\nu}k^\mu k^\nu\geq 0$ ---as it happens in GR if the null convergence condition holds. Even though some of the results can be derived even for negative $W$, from now on we will concentrate in the outstanding case with
\be\label{W>0}
W\geq 0
\ee
everywhere on $S$. Notice that this is still compatible with $G_{\mu\nu}k^\mu k^\nu <0$ as long as $\Sigma^k_{AB} \Sigma^k{}^{AB}\geq - G_{\mu\nu}k^\mu k^\nu$.
Following \cite{Snew2} the following terminology will be used
\be\label{Wset}
{\cal W}(S) := \{q\in S, \hspace{2mm} W|_q =0\}.
\ee

The function $W$ is crucial because, as can be seen from \eqref{Ln}, the stability operators are independent of the direction $\vec n$ if $W\equiv 0$, ergo the stability is ruled in that case  by the operator $L_{-\ell}$ (case with $n^\mu n_\mu =0$). Note that in general
$$
L_n =L_{-\ell} -(n^\mu n_\mu) W/2 .
$$
From here one notices that, under the assumption \eqref{W>0}, whenever $n^\rho n_\rho$ is positive (respectively negative) on $S$ the corresponding eigenvalue $\lambda_n$ will be smaller (resp.\  larger) than $\lambda_{-\ell}$. 
The intuitive idea to keep in mind is that --as long as \eqref{W>0} holds-- stability increases when pointing further `down' with $\vec n$, that is to say, when its norm $n^\mu n_\mu$ is smaller. This was discussed in detail in \cite{AMS,AMS1} and Result 2.4 in \cite{Snew2}. In particular, if \eqref{W>0} holds with $W\not\equiv 0$, a necessary condition for the stability of MOTSs along a {\em spacelike} $\vec n$ is its strict stability along the null direction $-\vec \ell$, that is, $\lambda_{-\ell}> 0$: a MOTS $S$ with $S\backslash {\cal W}(S)\neq \emptyset$ that is stable along a spacelike $\vec n$ must be strictly stable along $-\vec\ell$. 

These properties, together with the fundamental Result \ref{areaell}, allows one to prove the same conclusions as in Result \ref{areaell} for any external normal direction $\vec n$ that is non-timelike everywhere --or non-timelike outside ${\cal W}(S)$ \cite{Snew2}. Nevertheless, these are non-optimal bounds because under \eqref{W>0} one has
$$
\lambda_n \leq \lambda_{-\ell}
$$ 
and thus the optimal area bound is that given in \eqref{ellineq}, in the sense that the bound along the null direction $-\ell$ is stricter (or the same) than for any other choice of external normal vector that is non-timelike outside ${\cal W}(S)$. And similarly for the other points in Result \ref{areaell}. Nevertheless, notice that exactly the same conclusions as in Result \ref{areaell} hold for any normal that coincides with $-\vec\ell$ {\em just outside} ${\cal W}(S)$.

\begin{remark}\label{slice}
As mentioned in the Introduction, bounds using $\mu_S$ were presented in \cite{G,GM} for MOTSs stable within an initial data set, that is to say, embedded in a spacelike hypersurface. However, the restriction of having the MOTS in the initial slice implies that sometimes the area bound will not be optimal, as only the principal eigenvalue $\lambda_n$ for the normal $\vec n$ which is tangent to the initial slice is considered. An extreme case of this failure arises when the MOTS is {\em unstable} within the initial data set, as then the results in \cite{G,GM} do not hold, despite the fact that bounds for the area exist if $\lambda_{-\ell} \geq 0$ ---or even when only $\mu_S+\lambda_{-\ell} >0$.
\end{remark}

\section{The condition $\mu_S>0$ in General Relativity}
All the previous results are independent of any field equations, and they hold in any theory based on a 4-dimensional Lorentzian manifold. Nevertheless, they have a special relevance and applicability in the outstanding case of GR, where the condition $\mu_S>0$ can be put in correspondence with physical quantities such as the cosmological constant $\Lambda$, the energy density and pressure of fluids, or the electromagnetic charges, etc.

To start with, recall that in GR \eqref{efe} and \eqref{muinT} hold and thus $\mu_S >0$ is ensured when $\Lambda >0$ if the DEC holds. This is the case treated in \cite{Snew,Snew2} leading to the ultra-massive spacetimes discussed therein. However, $\mu_S>0$ is still compatible with all signs of $\Lambda$, and thus similar extreme spacetimes may appear in cases with vanishing $\Lambda$, or even with $\Lambda<0$, as long as the MOTS lies within matter with enough strength. To be specific, let me consider several prototypes of energy-momentum tensor separately.

\begin{itemize}
\item The energy-momentum tensor of the electromagnetic field $F_{\mu\nu}$
$$
T^{(em)}_{\mu\nu} = F_{\mu\rho} F_{\nu}{}^\rho -\frac{1}{4} F_{\rho\sigma} F^{\rho\sigma} g_{\mu\nu}= \frac{1}{2} \left(F_{\mu\rho} F_{\nu}{}^\rho+\stackrel{*}{F}_{\mu\rho} \stackrel{*}{F}_{\nu}{}^\rho \right)
$$
is known to satisfy the DEC, and thus ${\cal T}^{(em)}_{\mu\nu}\ell^\mu k^\nu$ is non-negative for sure. Here 
$$
\stackrel{*}{F}_{\mu\nu} :=\frac{1}{2} \eta_{\mu\nu\rho\sigma} F^{\rho\sigma}
$$
is the Hodge dual of $F_{\mu\nu}$ with $\eta_{\mu\nu\rho\sigma}$ the volume element 4-form.
An explicit computation provides
$$
T^{(em)}_{\mu\nu}\ell^\mu k^\nu = \frac{1}{2} \left((F_{\rho\sigma}\ell^\rho k^\sigma)^2 + (\stackrel{*}{F}_{\rho\sigma} \ell^\rho k^\sigma)^2 \right)=\frac{1}{2}  \left((F_{\rho\sigma}\ell^\rho k^\sigma)^2 +(F_{\rho\sigma}\epsilon^{\rho\sigma})^2\right)
$$
where $\epsilon_{\rho\sigma}$ is the volume element 2-form of $S$. These expressions are obviously non-negative. Furthermore, each summand represents the electric and magnetic charge density squared. Therefore, the contribution of any (non-vanishing) electromagnetic field to $\mu_S$ is always non-negative in GR and given by
\be\label{muelec}
\mu_S^{(em)} =\frac{4\pi G}{c^4} \min_S   \left((F_{\rho\sigma}\ell^\rho k^\sigma)^2 +(F_{\rho\sigma}\epsilon^{\rho\sigma})^2\right) >0.
\ee
\item The energy-momentum tensor of a scalar field $\phi$  with potential $V(\phi)$ reads
$$
T^{(scal)}_{\mu\nu} = \nabla_\mu \phi \nabla_\nu \phi -\frac{1}{2} g^{\rho\sigma} \nabla_\rho\phi \nabla_\sigma\phi\, g_{\mu\nu} -V(\phi) g_{\mu\nu}
$$
and this satisfies the DEC whenever $V(\phi) \geq 0$. In principle, the minimum absolute value $V_{\min}$ of $V(\phi)$ can be chosen at will, however, this is related to the problem (and the existence) of a cosmological constant with origin in a scalar field (such as the inflaton or the Higgs fields). Ultimately, the only important thing should be the effective cosmological constant that arises as the sum of $\Lambda$ with $8\pi G V_{min}/c^4$. A straightforward calculation gives
$$
T^{(scal)}_{\mu\nu} \ell^\mu k^\nu = V(\phi)+\frac{1}{2} f_\rho f^{\rho} \geq V(\phi)
$$
where $f_\rho$ is the part of $\nabla_\rho \phi$ tangent to $S$: $f_\mu :=(\delta_\mu^\nu +\ell_\mu k^\nu +\ell^\nu k_\mu) \nabla_\nu \phi$, so that $f_\rho f^\rho \geq 0$. The contribution of any scalar field to $\mu_S$ is then
\be\label{muscal}
\mu_S^{(scal)} =\frac{4\pi G}{c^4}  \min_S  \left( 2V(\phi) + f_\rho f^\rho\right).
\ee
\item The energy momentum tensor of a perfect fluid with velocity vector field $u^\mu$ is
$$
T^{(pf)}_{\mu\nu} = (\rho +p) u_\mu u_\nu +p g_{\mu\nu}
$$
where $\rho$ is the proper energy density and $p$ the isotropic pressure of the perfect fluid. A simple calculation leads to
$$
T^{(pf)}_{\mu\nu} \ell^\mu k^\nu = \frac{1}{2}\left[\rho -p + (\rho+p) v_\rho v^\rho \right] 
$$
where $v_\mu:= (\delta_\mu^\nu +\ell_\mu k^\nu +\ell^\nu k_\mu)u_\nu$ is the part of $u_\mu$ tangent to $S$, and thus $v_\rho v^\rho \geq 0$. We observe that this quantity is always non-negative if $|p|\leq \rho >0$ (this corresponds to the dominant energy condition (DEC) for perfect fluids). Thus, the contribution of a perfect fluid to $\mu_S$ is
\be\label{mupf}
\mu_S^{(pf)} = \frac{4\pi G}{c^4} \min_S \left[\rho -p + (\rho+p) v_\rho v^\rho \right] .
\ee
\item The energy-momentum tensor of a pure radiation field (also called null dust) propagating along the null vector field $L^\mu$  is given by
$$
T^{(pr)}_{\mu\nu} =A L_\mu L_\nu, \hspace{1cm} L_\mu L^\mu =0
$$
for some function $A$. Therefore
$$
T^{pr}_{\mu\nu} \ell^\mu k^\nu = A (L_\mu \ell^\mu) (L_\nu k^\nu)
$$
and this is always non-negative as long as $A>0$. The contribution of a null dust to $\mu_S$ is
\be\label{mupr}
\mu_S^{(pr)} = \frac{8\pi G}{c^4} \min_S \left[A (L_\mu \ell^\mu) (L_\nu k^\nu) \right].
\ee
\end{itemize}
In summary, in GR the value of $\mu_S$ given by \eqref{muinT} has a lower bound related to physical quantities as follows
\be\label{mutotal}
\mu_S \geq  \Lambda + \mu^{(scal)}_S +\mu^{(em)}_S +\mu^{(pf)}_S +\mu^{(pr)}_S +\mu^{(others)}_S
\ee
where $\mu^{(others)}_S$ represents any other possible matter and the other summands are defined in \eqref{muelec}--\eqref{mupr}. Observe that a positive $\Lambda$ sets a universal limit on the area of MOTSs, and this was the prominent case analyzed in \cite{Snew,Snew2}. Notice also that the combination $\Lambda +8\pi G \min V(\phi)/c^4$ arises as an effective cosmological constant if there is a scalar field in the spacetime.

\section{Generalized ultra-massive spacetimes}\label{sec:main}
The results in section \ref{sec:bounds} lead to a generalization of the concept of ultra-massive spacetime introduced in \cite{Snew,Snew2}. In particular it allows for an arbitrary sign of $\Lambda$ in GR.
The general definition is as follows.

\begin{defi}[Generalized ultra-massive spacetimes]\label{ultraM}
A spacetime will be called ``generalized ultra-massive'' if it contains a smooth dynamical horizon foliated by MTSs $\{S_s\}$ ($s\in [0,\tau)$) with the following properties
\begin{enumerate}
\item \label{point1} All the foliating $S_s$  satisfy \eqref{W>0} with $W\not\equiv 0$ (so that $S_s\backslash {\cal W}(S_s)\neq \emptyset$ for all $s\in [0,\tau)$).
\item \label{point2} $\bar\mu >0$, where $\bar \mu :=${\rm inf}$(\mu_{S_s})$, $s\in[0,\tau)$. 
\item \label{point3} For any positive $\epsilon >0$, there exist values of $s\in [0,\tau)$ such that the corresponding foliating MTSs have areas larger than $4\pi/\bar\mu-\epsilon$.
\end{enumerate}
\end{defi}
Observe that in GR, and if the DEC holds $\bar \mu \geq \Lambda$, hence all generalized ultramassive spacetimes are ultramassive in that case (as the original definition requires $\Lambda >0$). However, other cases may occur in alternative theories and for other values of $\Lambda$, depending also on the particular energy-momentum content. 
From point \ref{point2} $\mu_{S_s} >0$ for all $s\in [0,\tau)$, which implies that $S_s$ are topological spheres due to point \ref{topolnew} in Result \ref{areaell} because the foliating MTSs are necessarily stable along the {\em spacelike} direction tangent to the dynamical horizon, and this requires that $S_s$ are strictly stable along $-\vec\ell$, that is, $\lambda_{-\ell}>0$. Assumption \ref{point3}, together with the uniqueness of the foliation by MTSs \cite{AG} and the area law\footnote{The area law states that the total area of the MTSs foliating any MTT
are everywhere non-decreasing along external directions, see Result 2.9 in \cite{Snew2}} for dynamical horizons \cite{Hay,AK,BE,BE1,Snew2}, implies that the areas of the marginally trapped spheres foliating the dynamical horizon approach the bound value $4\pi/\bar\mu$ as much as desired.

\begin{theorem}
Let the spacetime be generalized ultra-massive in the sense of Definition \ref{ultraM} with a continuous Einstein tensor $G_{\mu\nu}$ and assume that the MTSs $\{S_s\}$
converge as $s\rightarrow \tau$ to a smooth closed MOTS that is not extremal and has $W|_{S_\tau}\not\equiv \emptyset$. 
Then
\begin{itemize}
\item The dynamical horizon is part of an MTT ${\cal H}$ that satisfies the area law with topology $\mathbb{R}\times \mathbb{S}^2$.
\item There is a distinguished MTS $\bar S := S_\tau \in {\cal H}$ with constant Gaussian curvature ${\cal K} _{\bar S}=\bar\mu$ ---and area $4\pi /\bar\mu$. 
\item All MTTs passing through $S_0$ (that necessarily weave each other, see \cite{AG,BeS}) change signature somewhere at $\bar S$. At least one such MTT ${\cal H}$ is null everywhere on $\bar S$.
\item If $\bar S$ is not extremal and $\bar S\setminus {\cal W}(\bar S)\neq \emptyset$, ${\cal H}$ becomes timelike towards the immediate past of $\bar S$. 
\end{itemize}
\end{theorem}
\proof Again the proof is basically the same as that for Theorem 3.1 in \cite{Snew2} with some needed adjustments. One can build a reference foliation composed by non-timelike hypersurfaces $\{\Sigma_s\}$ adapted to the dynamical horizon of Definition \ref{ultraM} in the sense that each $S_s$ is strictly stable within $\Sigma_s$, that is, along the normal $\vec n_s$ of $S_s$ that is tangent to $\Sigma_s$. From the results in \cite{AMS1} such a dynamical horizon extends to $s=\tau$ as an MTT whenever $S_\tau :=\bar S$ is not extremal. Such MTT will be part of ${\cal H}$ and is tangent to $\Sigma_\tau$ everywhere on $\bar S$ if $\bar S$ is just stable (but not strictly stable) there. 

Point \ref{ineqnew} in Result \ref{areaell} implies that 
$$(\mu_{S_s} +\lambda_{-\ell_s})A_{S_s}\leq 4\pi, \hspace{3mm} \forall s\in [0,\tau).
$$
where $\lambda_{-\ell_s}$ denote the principal eigenvalue of each $S_s$ along the corresponding null external direction $-\vec \ell_s$. As the $S_s$ are stable along the spacelike external direction tangent to the dynamical horizon, they have to be strictly stable along $-\vec\ell_s$ as follows from the discussion in subsection \ref{subsec:W>0} and the assumption \ref{point1} in Definition \ref{ultraM}. Thus, $\lambda_{-\ell_s} >0$ for all $s\in [0,\tau)$ and one finds
 $$\mu_{S_s} A_{S_s} <4\pi $$ 
 for all $s\in [0,\tau)$. Using here the definition of $\bar\mu$, one can actually write
 $$
 \bar\mu A_{S_s} < 4 \pi , \hspace{1cm} \forall s\in [0,\tau).
 $$
As the areas increase monotonically with $s$ due to the area law and the convergence to $\bar S=S_\tau$ , it follows that, in the limit $s\rightarrow \tau$,
$$
A_{\bar S} = 4\pi/\bar\mu
$$
otherwise a contradiction with assumption \ref{point3} in Definition \ref{ultraM} would arise. 
Now, one needs to check that $\mu_{\bar S} =\bar\mu$, but this follows from the continuity of the Einstein tensor and of the dynamical horizon. Hence, we are in the situation of point \ref{Aequalnew} in Result \ref{areaell} and we obtain
$$\lambda_{-\bar\ell}=0$$
where $\vec{\bar\ell}$ is the null $\ell$-direction on $\bar S$, the one with negative expansion (so that $\bar S$ is not extremal). Furthermore, $\bar S$ is a round sphere with ${\cal K} =\bar\mu$ and the MTT is tangent to $\Sigma_\tau$ everywhere and null (tangent to $\vec{\bar\ell}$) at $\bar S$. This implies that the MTT changes signature precisely at $\bar S$.

The rest of the argument is as in the proof of Theorem 3.1 in \cite{Snew2}.
\square

The round sphere $\bar S$ has several axial Killing vectors, but the angular momenta relative to these axial vectors all vanish. The proof is identical to that of Remark 3.2 in \cite{Snew2}.

\section{Discussion}\label{sec:discussion}
New area bounds have been presented for general (stable or not) MOTSs. The stability is ruled by the particular constant $\lambda_{-\ell}$ that determines the stability or not of the MOTS towards the past null direction with non-vanishing expansion, together with a particular component of the Einstein tensor on the MOTS. The results are independent of any field equations and of energy conditions. They acquire special relevance in the case of General Relativity. The existence of ultra-massive spacetimes, previously identified in the presence of a positive cosmological constant $\Lambda$  in GR, has been generalized to many cases with any possible value of $\Lambda$. 

Given that the properties of these generalized ultra-massive spacetimes are independent of the existence of symmetries, the results may have implications for mergers of, or accretion onto, very compact objects. For instance, immediately upon the merge of two black holes, the area of the newborn outermost MOTS is strictly larger than the sum of the areas of the two individual previous ones. As the system settles down the MOTS area continues to monotonically non-decrease, and thus an eventual stabilization into a final stationary Kerr black hole depends on the final value of $\mu_S$ (and in particular on $\Lambda$). In the known cases, the ultra-massive spacetimes do not have any event horizon nor future null infinity. Rather, they tend to present universal singularities, whose existence can be deduced from standard incompleteness theorems \cite{HE,P,S1,SMilestone} due to the existence of closed trapped surfaces, see also section 7 in \cite{AMMS}.

In this paper the case genus $g=0$ has been largely analyzed, probably the one with greatest relevance. However, for higher genus one can easily derive {\em lower} bounds for the area of the corresponding MOTS. Similarly, one can explore the relevance of the area bounds for the spherical case when just the combination $\mu_S +\lambda_{-\ell} $ is positive.

Finally, the spacetimes in which the area limits constraint marginally past-trapped surfaces are also interesting and worth studying.

\section*{Acknowledgments}
I am thankful to the group ``A new Geometry for Einstein's Theory of Relativity \& Beyond'' in the Faculty of Mathematics, Vienna University, for hospitality while this research was being finished. Research supported by Grant PID2021-123226NB-I00 funded by the Spanish
MCIN/AEI/10.13039/501100011033 together with “ERDF A way of making Europe” and by Spanish MICINN Project No. PID2021-126217NB-I00.


\begin{thebibliography}{[10]} 
\bibitem{AMMS} Andersson L, Mars M, Metzger J and Simon W, The time evolution of marginally trapped surfaces {\it Class. Quantum Grav.} {\bf 26} (2009) 085018

\bibitem{AMS} Andersson L, Mars M and Simon W 2005 Local existence of dynamical and trapping horizons {\it Phys. Rev. Lett} {\bf 95} 111102
\bibitem{AMS1} Andersson L, Mars M and Simon W 2008 Stability of marginally outer trapped surfaces and existence of marginally outer trapped tubes, {\it Adv. Theor. Math. Phys.} {\bf 12} 853 


\bibitem{AK} Ashtekar A and Krishnan B 2025, Quasi-local black hole horizons: recent advances {\it  Living Rev. Relativ.} {\bf 28} 8.
\bibitem{AG} Ashtekar A and Galloway G J 2005 Some uniqueness results for dynamical horizons {\it Adv. Theor. Math. Phys.} {\bf 9} 1



\bibitem{Bek} Bekenstein JD 1973 Black holes and entropy. {\it Phys. Rev. D} {\bf 7} 2333--2346.
\bibitem{Bek1} Bekenstein JD 1974 Generalized second law of thermodynamics in black hole physics. {\it Phys. Rev. D} {\bf  9} 3292--3300.
\bibitem{BeS} Bengtsson I and Senovilla J M M,  Region with trapped surfaces in spherical symmetry, its core, and their boundaries {\it Phys. Rev. D} {\bf 83} (2011) 044012


\bibitem{Booth} Booth I, Black hole boundaries {\it Can. J. Phys.} {\bf 83} (2005) 1073--1099

\bibitem{BBGV} Booth I, Brits L, Gonzalez J A and Van Den Broeck C, Marginally trapped tubes and dynamical horizons, {\it Class. Quantum Grav.} {\bf 23} (2006) 413

\bibitem{BE} Bousso R and Engelhardt N, New Area Law in General Relativity, {\it Phys. Rev. Lett.} {\bf 115} (2015) 081301

\bibitem{BE1} Bousso R and Engelhardt N, Proof of a new area law in general relativity, {\it Phys. Rev. Lett.} {\bf 92} (2015) 044031


\bibitem{G} Galloway, G.J. 2023 Remarks on the size of apparent horizons. {\it Lett. Math. Phys.} {\bf 113} 118.

\bibitem{GM} Galloway G J and Mendes A, 2018 Rigidity of marginally outer trapped 2-spheres,
{\it Comm. Anal. Geom.} {\bf 26} 63–83.

\bibitem{Haw} Hawking SW 1975 Particle creation by black holes. {\it Commun. Math. Phys.} {\bf 43} 199--220.

\bibitem{HE} S.~W. Hawking and G.~F.~R. Ellis. {\it The Large Scale Structure of Spacetime}. Cambridge University Press, 1973.

\bibitem{Hay} Hayward S A 1994 General laws of black-hole dynamics {\it Phys. Rev. D}, {\bf 49} 6467 

\bibitem{HSN} Hayward S A, Shiromizu T and Nakao K-i, A cosmological constant limits the size of black holes {\it Phys. Rev. D} {\bf 49} 1994 5080--85

\bibitem{JRD} Jaramillo J L, Reiris M and Dain S 2011 Black hole area-angular momentum inequality in non-vacuum spacetimes {\it Phys.Rev. D} {\bf 84} 121503(R) 


\bibitem{Mars} Mars, M. Stability of marginally outer trapped surfaces and applications, in {\em Recent Trends in Lorentzian Geometry}, S\'anchez, M., Ortega, M, and Romero, A. eds., (2012) 111-138 (Springer Proceedings in Mathematics \& Statistics, vol 26. Springer, New York, NY)


\bibitem{MaSe} Mars M and Senovilla J M M , Trapped surfaces and symmetries, {\it Class. Quantum Grav.} {\bf 20} (2003) L293


\bibitem{Ne} Newman R P A C, Topology and stability of marginal 2-surfaces, {\it Class. Quantum Grav.} {\bf 4} (1987) 277


\bibitem{P} Penrose R 1965 Gravitational collapse and space-time singularities {\it Phys. Rev. Lett.} {\bf 14}, 57


\bibitem{PBH} Pook-Kolb D, Booth I and Hennigar R A, Ultimate fate of apparent horizons during a binary black hole merger. II. The vanishing of apparent horizons, {\t Phys. Rev. D} {\bf 104} (2021) 084084 


\bibitem{S1} Senovilla J M M, Singularity theorems and their consequences, {\it Gen. Rel. Grav.} {\bf 30} (1998) 701--748

\bibitem{S0} Senovilla J M M Classification of spacelike surfaces in spacetime {\it Class. Quantum Grav.} {\bf 24} (2007) 3091--3124

\bibitem{S} Senovilla J M M, Trapped surfaces. {\it Int. J. Mod. Phys. D} {\bf  20} (2011) 2139--2168

\bibitem{SPrague} Senovilla J M M, On the stability operator for MOTS and the 'core' of Black Holes, in Proceedings of the Conference ``Relativity and Gravitation. 100 years after Einstein in Prague'', {\it Springer Proc. Phys.} {\bf 157} (2014) 215 (arXiv:1210.3731)

\bibitem{SERE} Senovilla J M M, Remarks on the stability operator for MOTS, in Proceedings of the Spanish Relativity Meeting in Portugal ERE2012, {\it Springer Proc. Math. Stat.} {\bf 60} (2014) 403 (arXiv:1211.6022)

\bibitem{Snew} Senovilla J M M, Ultra-massive spacetimes, Portugalia Math. {\bf 80} (2023) no.1/2, pp.133-155, DOI 10.4171/PM/2095 

\bibitem{Snew2} Senovilla J M M, Beyond black holes: Universal properties of 'ultra-massive' spacetimes, Class. Quantum Grav. {\bf 40}  (2023) 145002

\bibitem{SMilestone} Senovilla J M M and Garfinkle D, The 1965 Penrose singularity theorem, {\it Class. Quantum Grav.} {\bf 32} (2015) 124008


\bibitem{SILS} Shiromizu T, Izumi K, Lee K and Soligon D, Maximum size of black holes in our accelerating Universe,  Phys. Rev. D {\bf 106} (2022)  084014

\bibitem{Simon} Simon W 2012 Bounds on area and charge for marginally trapped surfaces with cosmological constant {\it Class. Quantum Grav} {\bf 29} 062001


\bibitem{W} Woolgar E, Bounded area theorems for higher-genus black holes {\it Class. Quantum Grav.} {\bf 16} (1999) 3005


\end{thebibliography}
\end{document}